\documentclass[11pt]{elsarticle}

\usepackage{hyperref}

\journal{Signal Processing}

\usepackage[T1]{fontenc}
\usepackage[utf8x]{inputenc}
\usepackage{graphicx,tikz,pgfplots}
\usepackage[cmex10]{amsmath}
\usepackage{amssymb,amsthm}
\usepackage{algorithmic}
\usepackage{array}
\usepackage{hyperref}

\usepackage{fixltx2e}

\newtheorem{lem}{Lemma}

\newtheorem{res}{Result}
\usepackage[compress,nospace]{cite}

\usepackage{color}
\usepackage{times}	
\definecolor{brickred}{rgb}{0.6,0,0}
\definecolor{byzantium}{rgb}{0.44,0.16,0.6}
\definecolor{darkgreen}{rgb}{0,0.6,0}

\begin{document}

\begin{frontmatter}

\title{Sparse-Based Estimation Performance for Partially Known  Overcomplete  Large-Systems }


\author[mymainaddress]{Guillaume Bouleux\corref{mycorrespondingauthor}}
\cortext[mycorrespondingauthor]{Corresponding author}
\ead{guillaume.bouleux@insa-lyon.fr}

\author[mysecondaryaddress]{R\'emy Boyer}
\ead{remy.boyer@l2s.centralesupelec.fr}

\address[mymainaddress]{Univ Lyon, INSA-Lyon, UJM-Saint-Etienne, DISP, EA 4570, 69621 Villeurbanne, France}
\address[mysecondaryaddress]{L2S laboratory (University of Paris-Sud, CentraleSupelec, CNRS), France}

\begin{abstract}
We assume the direct sum $\langle {\bf A}\rangle \oplus \langle {\bf B}\rangle$ for the signal subspace. As a result of post-measurement, a number of operational contexts presuppose the a priori knowledge of the $L_{B}$-dimensional "interfering" subspace $\langle {\bf B}\rangle$ and the goal is to estimate the $L_A$ amplitudes corresponding to subspace $\langle {\bf A}\rangle$. Taking into account the knowledge of the orthogonal "interfering" subspace $\langle {\bf B}\rangle\perp$, the Bayesian estimation lower bound is derived for the $L_A$-sparse vector in the doubly asymptotic scenario, i.e. $N, L_A, L_B \rightarrow \infty$ with a finite asymptotic ratio. By jointly exploiting the Compressed Sensing (CS) and the Random Matrix Theory (RMT) frameworks, closed-form expressions for the lower bound on the estimation of the non-zero entries of a sparse vector of interest are derived and studied. The derived closed-form expressions enjoy several interesting features: (i) a simple interpretable expression, (ii) a very low computational cost especially in the doubly asymptotic scenario, (iii) an accurate prediction of the mean-square-error (MSE) of popular sparse-based estimators and (iv) the lower bound remains true for any amplitudes vector priors. Finally, several idealized scenarios are compared to the derived bound for a common output signal-to-noise-ratio (SNR) which shows the interest of the joint estimation/rejection methodology derived herein.
\end{abstract}

\begin{keyword}
Overcomplete Bayesian linear model, asymptotic estimation performance,  subspace prior-knowledge, large-systems
\end{keyword}

\end{frontmatter}

\section{Introduction}

The  Compressive Sampling or Compressed Sensing (CS) is an attractive domain which gives new trends for people interested in sampling theory of sparse signals \citep{Baraniuk2008,Candes2005,Donoho2006}. The CS theory states that a sparse signal, \emph{i.e.}, a signal that can be decomposed as few non-zero values in a given basis (Fourier, wavelets,  etc.) can be sampled at a rate $T_S$ lower than the one predicted by the Shannon's theory. This paradigm has  been successfully exploited for solving ill-posed problems arising for instance in bio-medical analysis, RADAR detection, array processing,  wireless communications and radioastronomy imaging. In the CS framework, it is well known that any  matrix ${\bf H}$ of size $N \times L$ generated from an i.i.d. centered sub-Gaussian distribution with a variance of $1/N$ verifies the Restricted Isometry Property (RIP) \citep{Candes2005}  with a high probability \citep{Baraniuk2008}. On the other hand, the  doubly asymptotic spectrum and the empirical moments of the product ${\bf H}^T{\bf H}$  have been extensively studied  in the context of the Random Matrix Theory (RMT)  \citep{Couillet}. 

In the literature, CS and RMT techniques are usually applied to the noisy linear model where there is no  interfering signals. However, in a wide range of real life applications, the signal of interest is often corrupted by a partially known  interfering signal and an additive noise  (see \citep{Boyer2008,Bouleux2008,Bouleux2009,Wirfalt2012,Bouleux2013b}   for instance). This context motivates this work. More specifically, the CS and the RMT frameworks will be associated to derive new analytical closed-form expressions for the Bayesian lower bound  \citep{VanTrees2007}  on the estimation of a sparse amplitude vector \citep{Ollier2016} for the noisy linear model corrupted by a partially known  interfering signal.

\section{Compressed Sensing (CS) integrating an {\em a priori} knowledge}
\subsection{Definition of the CS model}\label{CSmodel}

Let ${\bf y}$ an observed vector of $N$ measurements corrupted by  an additive white centered zero-mean, Gaussian circular noise vector of variance $\sigma^2$. The standard CS model \citep{Candes2005,Baraniuk2008,Donoho2006} is defined according to 
\begin{eqnarray}\label{cs_model}
{\bf y} = {\bf \Psi} {\bf s} + {\bf n} = {\bf \Psi} {\bf \Phi}{\bf x}  + {\bf n}  
\end{eqnarray}
where ${\bf \Psi}$ is the known measurement matrix of size $N \times K$ with $N<K$, the vector ${\bf s}={\bf \Phi}{\bf x} $ of size $K\times 1$ admits an $L$-sparse representation, denoted by ${\bf x}$, in the basis ${\bf \Phi}$ (which could be Fourier basis, Wavelets basis, canonical basis, etc.) with $L<N$ and where ${\bf H} \stackrel{{\rm def.}}{=}  {\bf \Psi} {\bf \Phi}$ is often called the overcomplete dictionary.\\
One of the main  problems risen up by the theory of the Compressed Sensing relates to the minimum number of measurements $N$ needed for retrieving  the $L$-sparse vector $\bf{x}$. To address this problem, the authors of  \citep{Candes2005,Baraniuk2008,Donoho2006} have defined the Restricted Isometry Property (RIP). A standard strategy,  called universal design strategy to ensure that dictionary $\bf{H}$  satisfies the RIP condition with high probability, is to generate the {\rm i.i.d.} entries of dictionary matrix $\bf{H}$ following a sub-Gaussian distribution with zero mean and variance $1/N$ \citep{Baraniuk2008}.


\subsection{Exploiting the "interfering" subspace knowledge}
\label{sec:treatment}
 
\begin{figure}[htb]
\includegraphics[width=\columnwidth,height=9cm]{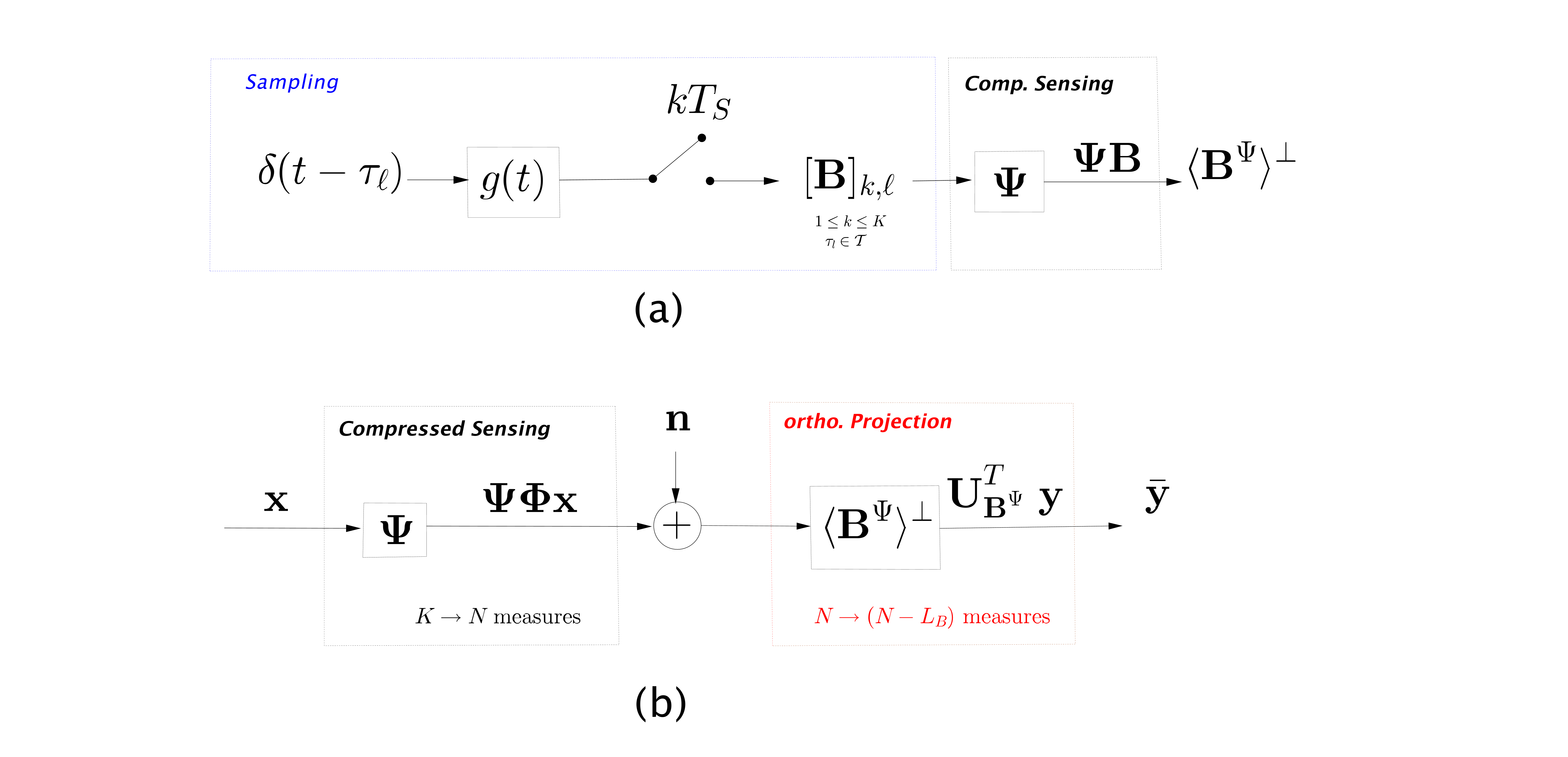}
\caption{$(a)$ Schematic construction of subspace $\langle {\bf B}^\psi \rangle$, $(b)$ Information processing scheme.}
\label{fig1}
\end{figure}

In many real life applications, we do have the knowledge of information given by the physics of the context. Those useful information help in tailoring models that precisely take into account the knowledge of particular frequencies \citep{Bouleux2013} for spectral analysis purpose, spatial angles for array processing \citep{Boyer2008} or RADAR  processing, and have demonstrated their power through biomedical analysis or radioastronomy imaging. So, we adopt the following "signal+interference" model $
{\bf s} = {\bf A} \boldsymbol{\alpha} + {\bf i} \ \mbox{with} \ {\bf i}={\bf B}  \boldsymbol{\beta}$
where $[{\bf A}]_{k,\ell} =  g(kT_S-\tau_\ell)$ with $1 \leq \ell \leq L_A$, $[{\bf B}]_{k,\ell'} =  g(kT_S-\tilde{\tau}_{\ell'})$ with $1 \leq' \ell \leq L_B$ are the "steering matrices" parametrized by the regular discretization at rate $T_S$ of a known waveform $g(t)$ along the time space. More precisely, $\mathcal{T}=\{\tau_\ell, 1\leq \ell \leq L_A\}$ stands for the time-delays of the $L_A$ sources of interest  $\boldsymbol{\alpha}$ and $\mathcal{\tilde{T}}=\{\tilde{\tau}_\ell, 1\leq \ell \leq L_B\}$ is  associated to the $L_B$ interfering sources $\boldsymbol{\beta}$. In the sequel, it is assumed that $(i)$ $\langle {\bf A}\rangle$ and $\langle {\bf B}\rangle$ are two disjoint subspaces, meaning that there is no time overlapping between  the sources of interest and of the interfering sources and $(ii)$ $\langle {\bf B}\rangle$ is known or previously estimated (matrix $\bf{A}$ and $\langle {\bf A}\rangle$ are unknown).  For instance,  the learning of  $\langle {\bf B}\rangle$ is based on pre-estimation of the clutter echo time-delays  in RADAR processing or by known strongly shining  "calibrator stars" in  radioastronomy imaging. The problem of interest  is to estimate vector $\boldsymbol{\alpha}$ based on a measurement vector where the contribution of $ {\bf i}$ has been removed using the knowledge of $\langle {\bf B}\rangle$. The standard "signal+interference" model described by signal ${\bf s}$ can be extended in the CS framework of model (\ref{cs_model}) following a straightforward strategy.  Let ${\bf \Phi}$ be a basis matrix such as $[{\boldsymbol{\Phi}}]_{k,k'} = g((k-k')T_S)$ where $1 \leq k,k' \leq K$. For a sufficiently fine partition, {\em i.e.},  for $K>N> L=L_A+L_B$, we have   $\langle {\bf A}\rangle \oplus \langle {\bf B}\rangle  \subset \langle {\bf \Phi}\rangle$.  Let ${\bf U}_{ {\bf B}^\psi  }$ be  a $N \times (N-L_{B})$  orthonormal basis matrix such as $\langle {\bf U}_{ {\bf B}^\psi  } \rangle = \langle {\bf B}^\psi \rangle^\perp$.  We have finally the deflated observation, defined according to
 \begin{eqnarray}\label{eq:obPCS}
{\bf \bar{y}} ={\bf U}_{ {\bf B}^\psi  }^T {\bf y} =  {\bf U}_{ {\bf B}^\psi  }^T {\bf H} {\bf x}+ {\bf \bar{n}}
\end{eqnarray}
 where  $ {\bf \bar{n}} = {\bf U}_{ {\bf B}^\psi  }^T{\bf n}$ and ${\bf x}$ is a $(K-L_A)$-sparse such as  ${\bf x}_\mathcal{T} = \boldsymbol{\alpha}$. The reader will find an illustration of the procedure in Fig. 1.


\section{ECRB for projected measurements and a large random dictionary}
\label{sec:perf}

\subsection{Dealing with projected measurements}

Let ${\rm MSE} =  \frac{1}{L_A}\mathbb{E}_{\bar{{\bf y}},\boldsymbol{\alpha}}\left[ \| \boldsymbol{\hat{\alpha}}({\bf \bar{y}}) - \boldsymbol{\alpha}\|^{2}\right] $ be the normalized Bayesian Mean Squared Error for an estimate $\boldsymbol{\hat{\alpha}}({\bf \bar{y}})$ of  $\boldsymbol{\alpha}$. The  Expected Cram\'er-Rao Bound (ECRB) \citep{VanTrees2007}, denoted by $C_{\mathbf{U}_{\mathbf{B}^{\psi}}^{T} \mathbf{A}^{\psi}}$ for the random amplitude vector $\boldsymbol{\alpha}$, of unspecified distribution $p(\boldsymbol{\alpha})$ given the observation model \eqref{eq:obPCS} fulfills  relation $
{\rm MSE}  \geq C_{\mathbf{U}_{\mathbf{B}^{\psi}}^{T} \mathbf{A}^{\psi}}=  \frac{\sigma^{2}}{L_A}\ {\rm Tr}\left\{ \left( \mathbf{A}^{\psi T} \mathbf{P}_{\mathbf{B^{\psi}}}^{\perp} \bf{A}^{\psi} \right)^{-1}\right\}$ where   ${\bf A}^\psi = \boldsymbol{\Psi} {\bf A}$. Introduce model $(\mathcal{M})$: $
 \bar{\bf{y}} | \boldsymbol{\alpha} \sim \mathcal{N}( \boldsymbol{\mu},\boldsymbol{\Sigma} ),$
where $\boldsymbol{\mu}={\bf U}_{ {\bf B}^\psi  }^T  {\bf A}^{\psi} \boldsymbol{\alpha}$ and $\boldsymbol{\Sigma} =\sigma^2 {\bf I}_{N-L_B} $  which is the covariance matrix of noise ${\bf \bar{n}}$. After some calculus, the ECRB admits the following expression:
\begin{align}
C_{\mathbf{U}_{\mathbf{B}^{\psi}}^{T} \mathbf{A}^{\psi}}&=   \frac{\sigma_\alpha^2}{  {\rm SNR^{(na)}}} \frac{{\rm Tr}\{ \mathbf{A}^{\psi T} \mathbf{P}_{\mathbf{B^{\psi}}}^{\perp} \bf{A}^{\psi}  \}}{ N-L_B} \frac{{\rm Tr}\left\{ \left( \mathbf{A}^{\psi T} \mathbf{P}_{\mathbf{B^{\psi}}}^{\perp} \bf{A}^{\psi} \right)^{-1}\right\}}{L_A} \label{eq:crbdef},
\end{align}
where  $
 {\rm SNR^{(na)}} = \frac{\mathbb{E} ||  \boldsymbol{\mu} ||^2}{{\rm Tr}\{\boldsymbol{\Sigma}\}} = \frac{\sigma_\alpha^2 {\rm Tr}\{ \mathbf{A}^{\psi T} \mathbf{P}_{\mathbf{B^{\psi}}}^{\perp} \bf{A}^{\psi}  \}}{\sigma^2(N-L_B)}$ is the output  and non-asymptotic SNR.

\subsection{Doubly asymptotic regime}
The practical interest of CRB-type expressions have been exposed in   \citep{Ben-Haim2010,Boyer2016} but we show in this work that expression \eqref{eq:crbdef} can be reduced to a very simple closed form expression  with the advantage of remaining valid even for the low sample regime, using some powerful results extracted  from the RMT  \citep{Couillet} where it is assumed   $N,L_A,L_B \rightarrow \infty$  with $N/L_A \rightarrow \rho$ and $L_B/L_A \rightarrow c$. Towards this goal, the following Lemma is provided.
\begin{lem}
\label{lem:1}
Let $\mathbf{F} = \mathbf{U}_{\mathbf{B}^{\psi}}^{T} \mathbf{A}^{\psi} \in \mathbb{R}^{(N-L_{B}) \times L_{A}}$ whose  elements $\left\{ F_{ij}\right\}_{i,j=1\ldots N-L_{B},L_{A}}$ are zero mean and i.i.d. with variance $\frac{1}{N}$.  Now, for $N,L_A,L_B \rightarrow \infty$,  and  $N/L_A \rightarrow \rho >1$, $(N-L_B)/L_A \rightarrow \tilde{\rho} = \rho -c >1$, then
 \begin{align}
 \frac{1}{L_{A}} {\rm Tr} \left\{ \left( \mathbf{F}^{T} \bf{F} \right)^{-1} \right\}& \overset{a.s.}{\longrightarrow} \frac{\rho}{\tilde{\rho}-1} = \frac{\rho}{\rho-c-1} \label{lem1:eq1},\\
 \frac{1}{N-L_{B}} {\rm Tr} \left\{  \mathbf{F}^{T} \bf{F}  \right\}& \overset{a.s.}{\longrightarrow} \frac{1}{\rho} \label{lem1:eq2},
\end{align}
where $a.s.$ stands for the almost sure convergence. 
\end{lem}
\begin{proof}
See the appendix.
\end{proof}
Under the assumptions of Lemma \ref{lem:1} and using \eqref{eq:crbdef}, a very compact expression of $C_{\mathbf{U}_{\mathbf{B}^{\psi}}^{T} \mathbf{A}^{\psi}}^\infty$ is enunciated by the following.
\begin{res}\label{res1}
 Assume that $N, L_A,L_B\rightarrow \infty$  and  $N/L_A \rightarrow \rho >1$, $(N-L_B)/L_A \rightarrow \tilde{\rho} >1$, then, we have $
C_{\mathbf{U}_{\mathbf{B}^{\psi}}^{T} \mathbf{A}^{\psi}}  \overset{a.s.}{\longrightarrow} C^{\infty}_{\mathbf{U}_{\mathbf{B}^{\psi}}^{T} \mathbf{A}^{\psi}}= \frac{\sigma_\alpha^{2}}{{\rm SNR}} \frac{1}{\tilde{\rho}-1}$
where ${\rm SNR} = \frac{\sigma_\alpha^2}{\sigma^2 \rho}$ is the almost sure doubly asymptotic equivalent of ${\rm SNR^{(na)}}$.
\end{res}
\section{Benchmarking ECRBs and  estimators}
This section is devoted to give a relation of order between the ECRB given by \eqref{eq:crbdef} with respect to two other ECRBs viewed as benchmarks and to analyze the behavior of sparse-based estimators. Let  $
(\mathcal{M}_0) : {\bf y}_0|\boldsymbol{\alpha} ,\boldsymbol{\beta} \sim \mathcal{N}\left(\mathbf{A}^{\psi}\boldsymbol{\alpha} +   \mathbf{B}^{\psi} \boldsymbol{\beta}, \sigma_0^2 {\bf I}_N \right)$ and 
$(\mathcal{M}_1) : {\bf y}_1|\boldsymbol{\alpha} \sim \mathcal{N}\left(\mathbf{A}^{\psi}\boldsymbol{\alpha}, \sigma_1^2 {\bf I}_N \right)$. Model $\mathcal{M}_0$ is  associated with the scenario where no ad-hoc strategy is developed to mitigate the corruption from the interference signals. In other words, the interference signals are wrongly interpreted as signals of interest. So, this bound does not solve the problem of interest and is given by
 \begin{align}
C_{[\mathbf{A}^{\psi}  \mathbf{B}^{\psi}]}& =   \frac{\sigma_0^{2}}{L}\ {\rm Tr}\left\{ \left( [\mathbf{A}^{\psi}  \mathbf{B}^{\psi}]^{T}  [\bf{A}^{\psi}  \bf{B}^{\psi}] \right)^{-1}\right\} \\ 
&= \frac{\sigma_\alpha^{2}}{  {\rm SNR_0^{(na)} } }\frac{{\rm Tr}\left\{ \left( [\mathbf{A}^{\psi}  \mathbf{B}^{\psi}]^{T}  [\bf{A}^{\psi}  \bf{B}^{\psi}] \right)^{-1}\right\}}{L}\nonumber\\ &  \cdot  \frac{\sigma_\alpha^{2}\ {\rm Tr}\left\{ \left( \mathbf{A}^{\psi T} \bf{A}^{\psi} \right)\right\} \ +\ \sigma_\beta^2 \ {\rm Tr}\left\{ \left( \mathbf{B}^{\psi T} \bf{B}^{\psi} \right)\right\}}{ N}  \label{eq:crbL},  
 \end{align}
where  ${\rm SNR_0^{(na)}} =  \frac{\sigma_\alpha^{2}\ {\rm Tr}\left\{ \left( \mathbf{A}^{\psi T} \bf{A}^{\psi} \right)\right\} \ +\ \sigma_\beta^2 \ {\rm Tr}\left\{ \left( \mathbf{B}^{\psi T} \bf{B}^{\psi} \right)\right\}}{\sigma^2_0 N}$ and ${\rm SIR}= \sigma_\alpha^{2}/\sigma_\beta^{2}$. The second model $\mathcal{M}_1$  is  associated with the ideal free-interference scenario. This bound admits the following expression:
 \begin{align}
C_{\mathbf{A}^{\psi}}  =  \frac{\sigma_1^{2}}{L_A}\ {\rm Tr}\left\{ \left( \mathbf{A}^{\psi T} \mathbf{A}^{\psi} \right)^{-1}\right\}
=  \frac{\sigma_\alpha^{2}}{  {\rm SNR_1^{(na)}}} \ \frac{{\rm Tr}\left\{ \left( \mathbf{A}^{\psi T} \mathbf{A}^{\psi} \right)\right\}}{N }\frac{{\rm Tr}\left\{ \left( \mathbf{A}^{\psi T} \mathbf{A}^{\psi} \right)^{-1}\right\}}{L_A}, \label{eq:crbLA}
 \end{align}
where
 ${\rm SNR_1^{(na)}} = \frac{\sigma_\alpha^{2} \ {\rm Tr}\left\{ \left( \mathbf{A}^{\psi T} \mathbf{A}^{\psi} \right)\right\}}{\sigma^{2}_{1} N}.$ Using a similar methodology as before, we have the following result given in the doubly asymptotic regime context.
\begin{res}\label{res2}
Assume that $N,L_A \to \infty$ with $N/L_A \rightarrow \rho >1$, $N/L \rightarrow \bar{\rho} >1$, then, we have $
C_{[\mathbf{A}^{\psi}  \mathbf{B}^{\psi}]} \overset{a.s.}{\longrightarrow} C_{[\mathbf{A}^{\psi}  \mathbf{B}^{\psi}]}^\infty = \frac{\sigma_\alpha^{2}}{{\rm SNR}_0} \left(\frac{1-{\rm SIR}^{-1}}{\rho}+ \frac{{\rm SIR}^{-1}}{\bar{\rho}}\right)\frac{\bar{\rho}}{\bar{\rho}-1}$ and $
C_{\mathbf{A}^{\psi}}  \overset{a.s.}{\longrightarrow}  C_{\mathbf{A}^{\psi}}^\infty = \frac{\sigma_\alpha^{2}}{{\rm SNR}_1} \frac{1}{\rho-1}$ where  ${\rm SNR}_0 = \frac{(\sigma_\alpha^{2}-\sigma_\beta^{2})}{\sigma_0^2 \rho} + \frac{\sigma_\beta^{2} }{\sigma_0^2 \bar{\rho}}$ and ${\rm SNR}_1 = \frac{\sigma_\alpha^2 }{\sigma_1^2 \rho}$ are the almost sure doubly asymptotic equivalent of $\;{\rm SNR_0^{(na)} }$ and $\;{\rm SNR_1^{(na)} }$, respectively.
\end{res}
\subsection{Simulation Tests}
We propose now different scenarios for (i) numerically showing the limit of the doubly asymptotic approximation and (ii) for analyzing the behavior of popular sparse-based estimators when they are tailored with the account of the interference signal. We have plotted on Fig. 2 the Mean Square Error of the approximation between analytic expressions \eqref{eq:crbdef}, \eqref{eq:crbL}, \eqref{eq:crbLA} of the ECRBs and their doubly asymptotic expressions (Results \ref{res1} and \ref{res2}) with respect to the number of samples. This figure leaves no room for doubt about the speed at which the asymptotic expressions for the ECRBs merge the theoretical ones. Finally, we have computed (Fig. 3) for severals growing dimensions of the $L_{B}$-dimensional subspace $\langle \bf{B} \rangle$ and six punctual  values of $L_{A}$, selected from the range $[0.01 N : 0.3 N]$, the ECRBs. Like Fig. 2, we clearly see that both plots merge whatever the values for $L_A$ and $L_B$.\\

\begin{figure}[htb]
\includegraphics[width=\columnwidth,height=9cm]{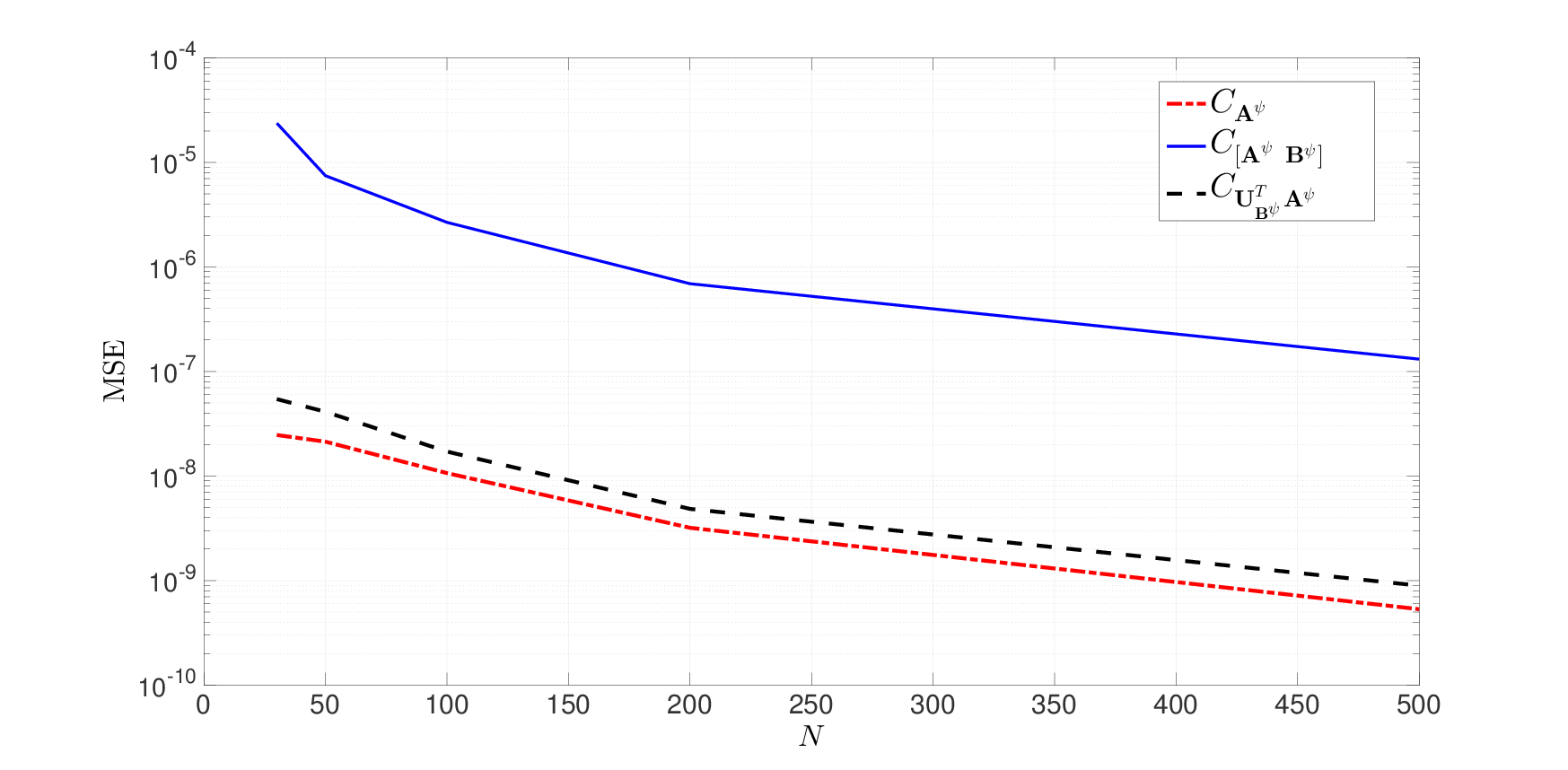}
 \caption{ MSE {\em Vs.} N for $L_A = L_B= N/10$, SNR=10 dB and $
\sigma_\alpha^2 = 1 $, $\sigma_\beta^2  = 10$}
 \label{fig2}
\end{figure}

\begin{figure}[htb]
\includegraphics[width=\columnwidth,height=9cm]{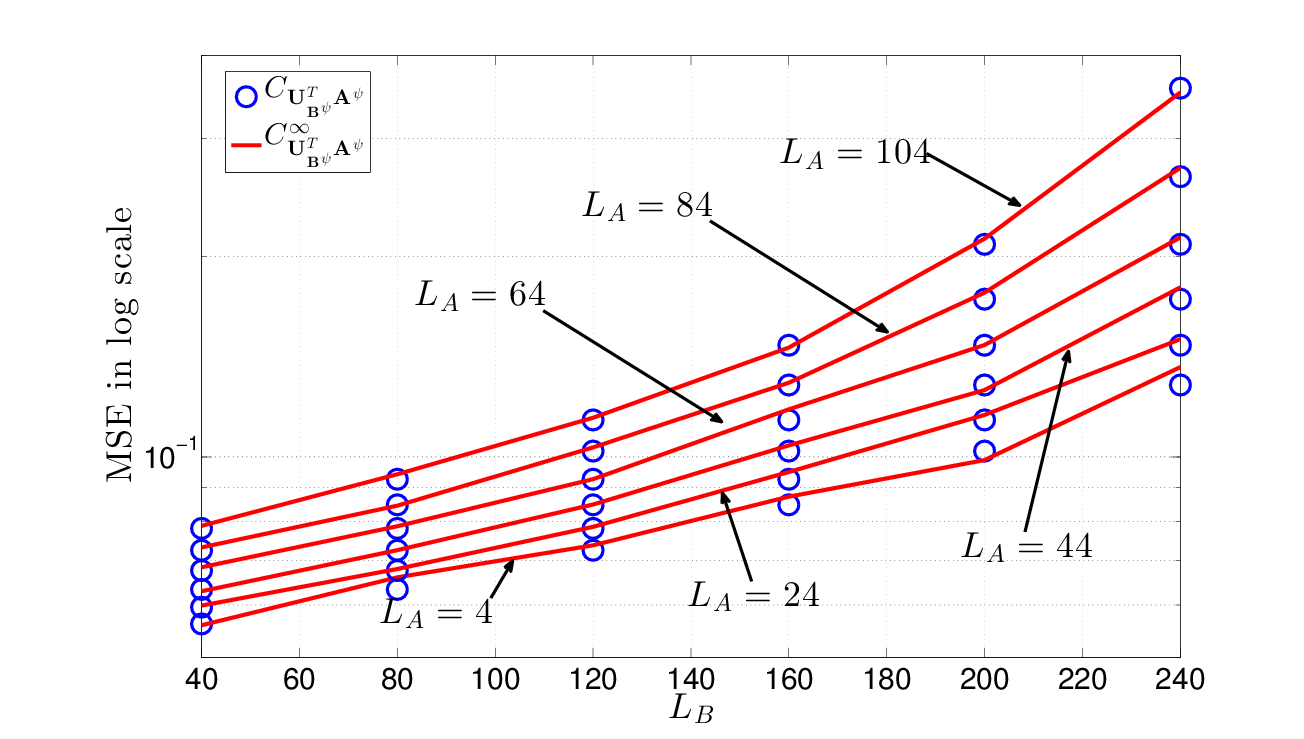}
 \caption{ MSE {\em Vs.} $L_B$  for different values of $L_A$, SNR =10 dB and $N = 100$, $\sigma_\alpha^2 = \sigma_\beta^2  = 1$}
 \label{fig3}
\end{figure}

We now turn the discussion to the behavior of the tailored estimators. To be sufficiently general, we have not designed the observation model with specific waveforms or delays but directly implemented dictionary $\mathbf{H}$ as presented by section \ref{CSmodel}. The Basis Pursuit DeNoise \citep{Chen1998}, the CoSaMP \citep{Needell2009} and the Orthogonal Matching Pursuit \citep{Tropp2007} sparse estimators have been computed with a deflated observation signal and a deflated dictionary for them to estimate the $L_{A}$ parameters of interest only. We have run 500  Monte-Carlo trials for each scenario with a common output SNR and confronted the Mean-Square-Error (MSE) with the ECRBs \eqref{eq:crbdef}, \eqref{eq:crbL}, \eqref{eq:crbLA} and their doubly asymptotic equivalents given in Results \ref{res1} and \ref{res2} for both scenarios drawn through Fig. 4. When the observation signal is composed by as much equally powered sources of interest as interfering ones (Fig. 4), we observe that $C_{\mathbf{U}_{\mathbf{B}^{T}} \mathbf{A}} $ is clearly next to $C_{\mathbf{A}} $.  Then, a quick analysis on both plots, reveals that none of the algorithms can perform well in a low SNR regime. The interfering signal has been almost rejected in that situation and only the OMP estimator seems to be in capacity to do  from a 20 dB SNR. The CoSaMP performs as well as the interfering signal was of interest, it is so unable to reject the influence of the interfering part and the BPDN does not reach any lower bounds. Curiously thereafter, when the signal of interest is buried into many interfering and much powerful elements (Fig. 5), we notice a significant gain (more than 15 dB) for each estimators when they account for the knowledge of the interferences. In that difficult scenario, the OMP still reaches  from a 20 dB SNR the lower bound $C_{\mathbf{U}_{\mathbf{B}^{T}} \mathbf{A}}$ which still remains close to the ideal bound $C_{\mathbf{A}} $ for which there is no interfering elements.

\begin{figure}[htb]
\includegraphics[width=\columnwidth,height=9cm]{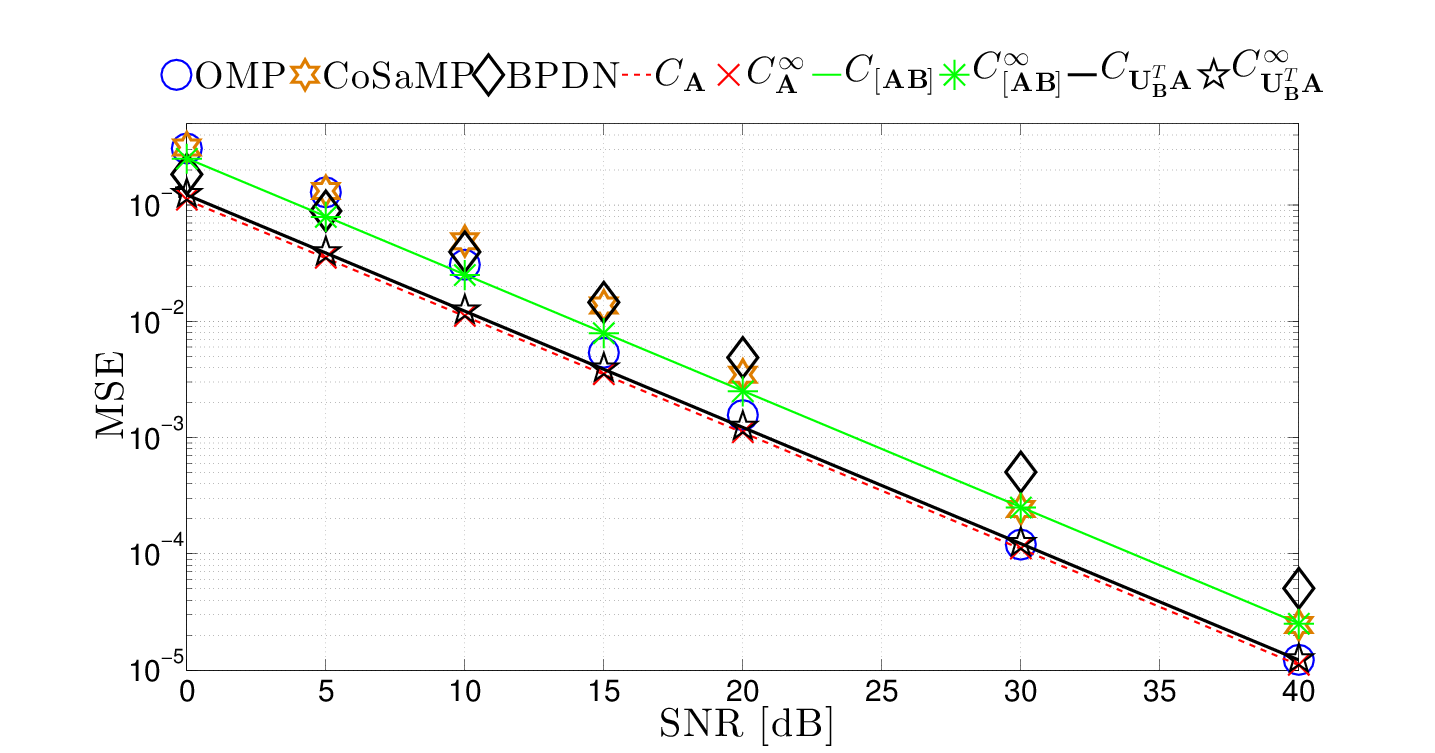}
 \caption{ MSE {\em Vs.} SNR in dB $L_A = L_B= 10$, $
N = 100$ and $
\sigma_\alpha^2 = \sigma_\beta^2  = 1$}
 \label{fig4}
\end{figure}

\begin{figure}[htb]
\includegraphics[width=\columnwidth,height=9cm]{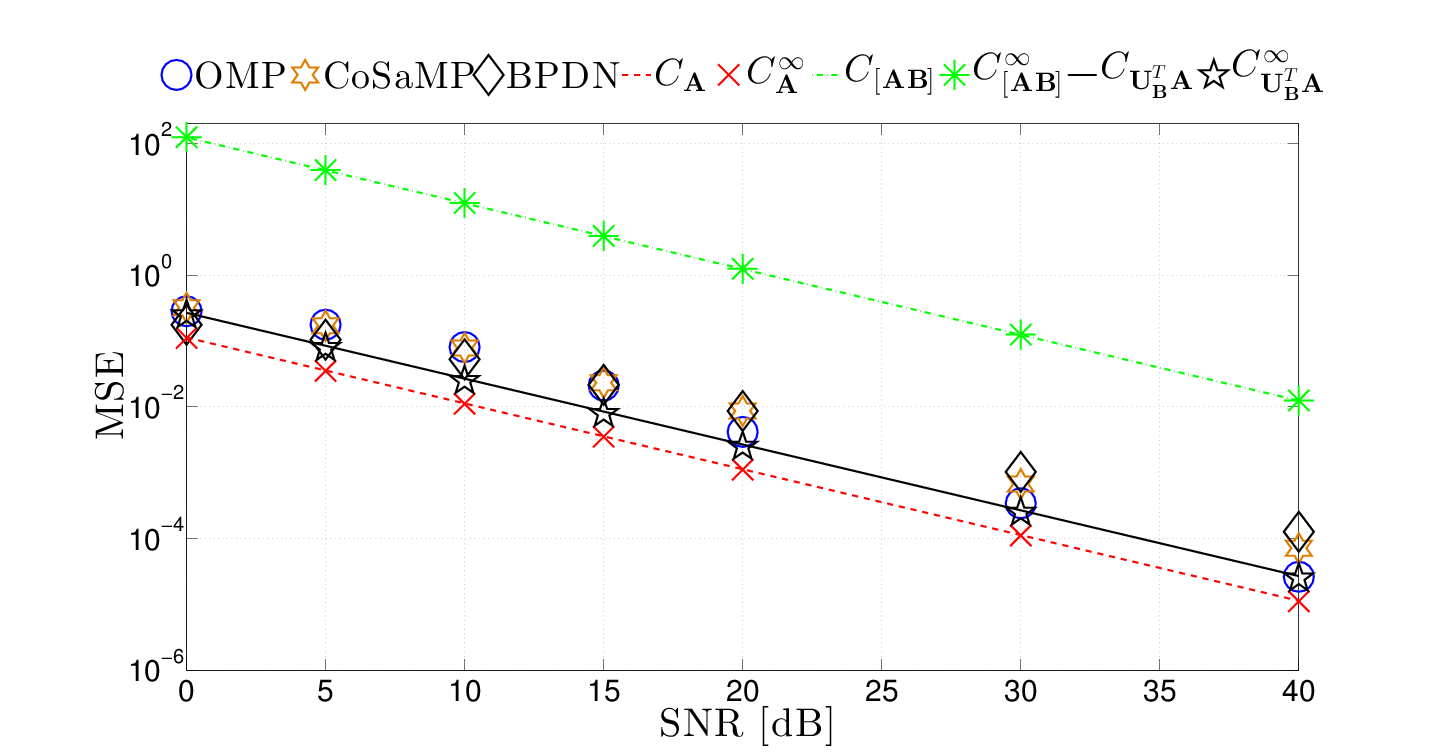}
 \caption{ MSE {\em Vs.} SNR in dB $L_A = 10$, $
L_B = 50$, $
N = 100$ and $
\sigma_\alpha^2 = 1$ and $
\sigma_\beta^2 = 100$}
 \label{fig5}
\end{figure}

\section{Conclusion}
In this work, the problem of interest is to derive the Bayesian performance bound for the estimation of the $L_A$ non-zero amplitudes of a sparse signal of interest based on the observation of a $N \times 1$ compressed measurement vector ${\bf y}$. Vector ${\bf y}$ follows an overcomplete Bayesian linear model corrupted by a set of interfering signals spanning an {\em a priori}  known $L_B$-dimensional subspace. Based on this standard assumption, the measurement vector ${\bf y}$ is confined in a $N-L_B$ subspace thanks to an orthogonal  deflation technique. In addition, the proposed analysis is done in the asymptotic framework, {\em i.e.},  for  $N,L_A,L_B \rightarrow \infty$  with finite asymptotic ratios. Our methodology allows to obtain $(i)$ an  easily interpretable expression of the bound, $(ii)$ a cheap computational cost especially in the doubly asymptotic scenario, $(iii)$ an accurate prediction of the mean-square-error (MSE) of popular sparse-based estimators and $(iv)$ a lower bound  for any amplitudes vector priors. Finally, several idealized scenarios are compared to the derived bound for a common output signal-to-noise-ratio (SNR) which shows the interest of our joint estimation/rejection methodology.

\section{Appendix : Proof of Lemma 1}

We first need to introduce the normalized trace of the resolvent for the random form  $\rho \mathbf{F}^{T} \bf{F}$  by the complex function $
h(z) = \frac{1}{L_{A}} {\rm Tr}\left\{ \left(\rho \mathbf{F}^{T}\mathbf{F}-z \mathbf{I}\right)^{-1}\right\}.$
Owing to the assumptions stated by Lemma \ref{lem:1} and due to  \citep{Marchenko1967}, we  know that the  empirical distribution of  $\rho \mathbf{F}^{T} \bf{F}$ defined by  
$
\hat{\mu}_{ \rho \mathbf{F}^{T} \bf{F}} = \frac{1}{L_{A}} \sum_{i=1}^{L_{A}} \delta_{\lambda_{i}( \rho \mathbf{F}^{T} \bf{F}) } $
with $\delta_{\lambda_{i}}$ the Dirac measure for eigenvalue $\lambda_{i}$, that is the unique probability measure satisfying $f(x) \delta_{\lambda} = f(\lambda)$ for any continuous function $f \in \mathbb{R}$; converges almost surely  in distribution towards a deterministic distribution function $\mu_{mp}$, \emph{i.e.} $\hat{\mu}_{ \rho \mathbf{F}^{T} \bf{F}} \overset{a.s.}{\longrightarrow}  \mu_{mp}$. The deterministic distribution $\mu_{mp}$ is supported on the compact interval $[\lambda_{-}, \lambda_{+}]$ possibly with a point mass at 0 \citep{Couillet} with  generalized density $
\frac{d\mu_{mp}(x)}{dx}  =  \frac{\sqrt{\left( \lambda_{+} -x \right) \left( x - \lambda_{-} \right)}}{2\pi \tilde{\rho} x}
$ with $
\lambda_{\pm}  =  (1 \pm \sqrt{\tilde{\rho}})^{2}$ and referred to as the Marchenko-Pastur density \citep{Marchenko1967}. For any distribution function $\mu$, the  Stieltjes Transform of $\mu$ denoted by $S_{\mu} (z) : \mathbb{C} \; \textbackslash$ Supp$(\mu) \longrightarrow \mathbb{C} $ is defined as $
\label{eq:stieltjesTF}
S_{\mu} (z)= \int \frac{\mu(d\lambda)}{\lambda -z}. $ It has been already shown \citep{Pastur2011} that the Stieltjes Transform $S_{\mu_{mp}}(z)$ of $\frac{d\mu_{mp}(x)}{dx}$ respects the following quadratic relation 
\begin{equation}
S_{\mu_{mp}}(z) = \frac{-1}{z} + \frac{\tilde{\rho}}{z}\frac{ S_{\mu_{mp}}(z)}{\left(1+ S_{\mu_{mp}}(z)\right)}\end{equation}
and proved to converge  \citep{Pastur2011}, towards $h(z)$. Notice now that the Stieltjes Transform $S_{\mu_{mp}}(z)$ weighted by $\frac{L_{A}}{N} \rightarrow\rho^{-1}$ and evaluated at $z=0$ corresponds to $
\frac{1}{L_{A}} {\rm Tr}\left\{ \left(\mathbf{F}^{T}\bf{F}\right)^{-1}\right\}$. With all these materials, \eqref{lem1:eq1} of lemma \ref{lem:1} is obtained by expressing the weighted Stieltjes Transform $ S_{\mu_{mp}}(0)$ through the above expression. To prove \eqref{lem1:eq2} we use the spectral theorem applied to empirical distribution $\hat{\mu}_{\rho \bf{F}^{T} \bf{F}}$, to obtain 
\begin{eqnarray}
\int f(x) \; \hat{\mu}_{\mathbf{F}^{T} \bf{F}}(dx)  &= &\frac{1}{L_{A}} \sum_{i=1}^{L_{A}} f(\lambda_{i}) =  \frac{1}{L_{A}} {\rm Tr}\left\{ f \left(  \mathbf{F}^{T} \bf{F} \right)\right\}.
\end{eqnarray} 

Since this empirical distribution converges when weighted by $\rho$ towards the Marchenko-Pastur distribution, \emph{i.e.} $\hat{\mu}_{ \rho \mathbf{F}^{T} \bf{F}} \overset{a.s.}{\longrightarrow}  \mu_{mp}$, we have consequently for $f$ a polynomial function the assertion 
\begin{eqnarray}
\frac{1}{L_{A}} {\rm Tr} \left\{ \left(\rho \mathbf{F}^{T} \mathbf{F} \right)^{k} \right\} & \overset{a.s.} {\longrightarrow} & \int_{\lambda_{-}}^{\lambda_{+}} x^{k} \mu_{mp}(dx) = \sum_{i=1}^{k} \frac{1}{k} \begin{pmatrix} k\\ i \end{pmatrix} \begin{pmatrix} k \\ i-1 \end{pmatrix} \tilde{\rho}
\end{eqnarray}

which reduces obviously to $\tilde{\rho}$ for $k = 1$. We have consequently $\frac{1}{L_{A}} {\rm Tr} \left\{ \left(  \mathbf{F}^{T}\bf{F} \right)^{k}\right\} \overset{a.s.}{\longrightarrow}  \rho^{-1} \tilde{\rho}$ and after obvious manipulations \eqref{lem1:eq2} is reached.

\section*{References}
\bibliographystyle{elsarticle-num}
\bibliography{Bibletter.bib}

\end{document}